\documentclass[11pt]{article}

\usepackage[ruled,vlined]{algorithm2e}
\usepackage{amsmath}
\usepackage{amssymb}
\usepackage{balance}
\usepackage{graphicx, subfigure}
\usepackage{times,color}
\usepackage{ifthen}

\newtheorem{theorem}{Theorem}[section]
\newtheorem{lemma}[theorem]{Lemma}

\newtheorem{corollary}[theorem]{Corollary}

\newcommand{\sq}{\hbox{\rlap{$\sqcap$}$\sqcup$}}
\newcommand{\qed}{\hspace*{\fill}\sq}
\newenvironment{proof}{\noindent {\bf Proof.}\ }{\qed\par\vskip 4mm\par}

\allowdisplaybreaks[1]

\newcommand{\floor}[1]{\lfloor #1 \rfloor}
\newcommand{\ceil}[1]{\lceil #1 \rceil}

\newcommand{\opt}{\mathrm{OPT}}
\newcommand{\alg}{\mathrm{ALG}}

\newcommand{\MyFrame}[1]{\noindent \framebox[\textwidth]{ \begin{minipage}{0.97\textwidth} #1 \end{minipage}}}%

\begin{document}
\begin{titlepage}
\title{TSP with Time Windows and Service Time%
\protect\footnote{Supported in part by the Israel Science Foundation
(grant No. 1404/10), by the Google Inter-university center and by
The Israeli Centers of Research Excellence (I-CORE) program, (Center
No.4/11).}}

\author{ Yossi Azar \thanks{School of Computer Science, Tel-Aviv University, Tel-Aviv 69978, Israel. E-mail: {\tt  azar@tau.ac.il} } \and Adi Vardi
\thanks{School of Computer Science, Tel-Aviv University, Tel-Aviv 69978, Israel.
 E-mail: {\tt  adi.vardi@gmail.com}}
 }

\maketitle

\vspace{-0.1cm}
\begin{abstract}
We consider TSP with time windows and service time. In this problem
we receive a sequence of requests for a service at nodes in a metric
space and a time window for each request. The goal of the online
algorithm is to maximize the number of requests served during their
time window. The time to traverse an edge is the distance between
the incident nodes of that edge. Serving a request requires unit
time. We characterize the competitive ratio for each metric space
separately. The competitive ratio depends on the relation between
the minimum laxity (the minimum length of a time window) and the
diameter of the metric space. Specifically, there is a constant
competitive algorithm depending whether the laxity is larger or
smaller than the diameter. In addition, we characterize the rate of
convergence of the competitive ratio to $1$ as the laxity increases.
Specifically, we provide a matching lower and upper bounds depending
on the ratio between the laxity and the TSP of the metric space (the
minimum distance to traverse all nodes). An application of our
result improves the lower bound for colored packets with transition
cost and matches the upper bound. In proving our lower bounds we use
an interesting non-standard embedding with some special properties.
This embedding may be interesting by its own.
\end{abstract}

\thispagestyle{empty}
\end{titlepage}

\section{Introduction}
Consider an employee in Google IT division. He his responsible for
replacing malfunctioning disks in Google huge computer farms. During
his shift he receives requests to replace disks at some points in
time. Each request is associate with a deadline. If the disk will
not be replaced before the deadline expiration, there is a high
probability to a significant hit in the performance of the Search
Engine. Replacing a disk takes a constant time (service time).
However, before the employee can replace it, he must travel from his
current location to the location of the disk. The goal is to
maximize the number of disks replaced before their deadline expired.
What path should the employee take and how the path changes with new
requests? We call this problem TSP with time windows and service
time (or vehicle routing with time windows and production costs).
All the requests are of unit service time, and the window of request
$i$ is $[r_i,d_i]$. In this paper we characterize the competitive
ratio for each metric space separately. We determine whether the
competitive ratio is constant or not depending on the minimum laxity
(the minimum length of a time window) and the diameter of the metric
space (the maximum distance between nodes in the metric space). In
addition, we consider the case where the laxity is large.
Specifically, we provide a matching lower and upper bounds depending
on the ratio between the laxity and the TSP of the metric space (the
minimum distance to traverse all nodes).

Note that if the service time is negligible compared to the minimum
positive distance between nodes then the problem reduces to the TSP
(or vehicle routing) with time windows
\cite{bansal2004approximation}. Moreover, if in addition all the
deadlines are the same and all the release times are zero then the
problem reduces to the well known (offline) orienteering problem
\cite{arkin1998resource,awerbuch1998new,golden1987orienteering}.

We note that even when the service time is not negligible, the TSP
with time windows and service time can be reduced to TSP with time
windows \cite{bansal2004approximation} by changing the metric space.
However, our competitive ratio depends on the properties of the
metric space. The reduction might change the parameters of the
metric space significantly. Hence, it might influence the crucial
parameter which determine the competitive ratio. Therefore, we
maintain the service time in our model.

Vehicle routing problems (with time windows) has been extensively
studied both in computer science and operation research literature,
see
\cite{desrochers1988vehicle,desrochers1992new,savelsbergh1985local,thangiah1993vehicle,tan2001heuristic,nagamochi2008approximating}.
For an arbitrary metric space Bansal et al.
\cite{bansal2004approximation} showed $O(\log^2 n)$-approximation
(for certain cases a better approximation can be achieved
\cite{chekuri2007approximation}). Constant factor approximations
have been presented for the case of points on a line
\cite{bar2005approximating,tsitsiklis1992special,karuno20032}. For
the orienteering problem, i.e., all the release times are zero and
all the deadlines are the same, there are constant factor
approximation algorithms
\cite{bansal2004approximation,chekuri2012improved,chekuri2004maximum,blum2003approximation}.

To the best of our knowledge we are the first to consider the online
version of this general problem.

Another motivation for our problem is the colored Packets with
Deadlines and Metric Space Transition Cost. In this problem we are
given a sequence of incoming colored packets. Each colored packet is
of unit size and has a deadline. There is a reconfiguration cost
(setup cost) to switch between colors (the cost depends on the
colors). The goal is to find a schedule that maximizes the number of
packets that were transmitted before the deadline. Note that for one
color the earliest deadline first (EDF) strategy is known to achieve
an optimal throughput. The unit cost color has been considered in
\cite{Azar_BMC}. In particular, an application of our result to the
uniform metric space, improves their lower bound and matches their
upper bound.

\subsection{Our results} \label{subsec:OurResults}
Let $L = \min_{i \in \sigma} \{d_i - r_i\} \geq 1$ denote the
minimum laxity of the requests (the minimum length of a time
window). Let $\Delta(G)$ be the diameter of the metric space $G$,
i.e., the largest distance between two nodes. Let TSP(G) denote the
weight of the minimal Traveling Salesperson Problem (TSP) in the
metric space $G$. Let MST(G) denote the weight of the minimal
spanning tree (MST) in the metric space G.

In this paper we characterize when it is possible to achieve a
$\Theta(1)$ competitive algorithm and when the best competitive
algorithm is unbounded. Moreover, we characterize the rate of
convergence of the competitive ratio to $1$ as the laxity increases.
Specifically, we provide a matching lower and upper bounds depending
on the ratio between the laxity and the TSP of the metric space. We
consider three cases.

\begin{itemize}

\item Case A: $L < \Delta(G) / 2$. For any metric space
the competitive ratio of any online algorithm is unbounded. The
claim is easily proved.

\item Case B: $(2+\epsilon)\Delta(G) < L \leq TSP(G)$ for any $\epsilon > 0$.
We design $O(1)$ competitive algorithm and a $1.002$ lower bound.

\item Case C: $L > TSP(G)$.
Let $\delta = TSP(G) / L < 1$. We show a strictly larger than $1$
lower bound. Specifically, if $\delta \leq \frac{1}{9}$ we provide a
lower bound of $1 + \Omega\big(\sqrt{\delta}\big)$ as well as a
matching upper bound of $1 + O\big(\sqrt{\delta}\big)$. We note that
without service time it is easy to design $1$-competitive algorithm
by traveling over TSP periodically. Recall that there is a reduction
from the service time model to a model without service time that
seems to contradict the lower bound (see
\cite{bansal2004approximation}). Nevertheless, the reduction
modifies the metric space and hence increases $\delta$ such that
$\delta$ is not smaller than $\frac{1}{9}$.

\end{itemize}

Note that in the remaining cases, i.e., $\Delta(G) / 2 \leq L \leq
(2+\epsilon)\Delta(G)$ the question whether there exists a constant
competitive algorithm depends on the metric space. Specifically, for
$L = \Delta(G)$ it is easy to proof that there is no constant
competitive algorithm for the uniform metric space. In contrast,
there is a constant competitive algorithm for the line metric space.

Observe that when the metric space consists of a single node (i.e.,
no traveling time) the optimal algorithm to serve a requests is EDF
(earliest deadline first) which is $1$-competitive. This case is
equivalent to packet scheduling with deadlines. If packets have
colors and switching between colors costs $1$, then our result
improves the lower bound of \cite{Azar_BMC}. Specifically, we
improve their $1 + \Omega\big(\delta\big)$ lower bound to $1 +
\Omega\big(\sqrt{\delta}\big)$ (where $C$ is the number of nodes in
$G$ and hence, $\delta = \frac{C}{L} < 1$) and match their upper
bound for the uniform metric space.

It is also interesting to mention that in many cases the competitive
ratio of an algorithm is computed as a supremum over all the metric
spaces and the lower bound is proved for one specific metric space.
We prove more refined results. Specifically, we show an upper bound
and a lower bound for each metric space separately. Hence, one can
not design a better competitive algorithm for the specific metric
space that one encounter in the real specific instance.

{\bf Embedding result.} One of the technique that is used for the
lower bound is an embedding that is interesting on its own. Let
$w({\rm S})$ denote the weight of the star metric $S$ (i.e., the sum
of the weights of the edges of S). We prove that for any given
metric space ${\rm G}$ on nodes $V$ and for any vertex $v_0 \in V$
there exists a star metric ${\rm S}$ with leaves $V$ and an
embedding $f: {\rm G} \rightarrow {\rm S}$ from ${\rm G}$  to ${\rm
S}$ $(f$ depends on $v_0)$ such that:
\begin{enumerate}
\item $w({\rm S}) = {\rm MST(G)}$.

\item The weight of every Steiner tree in $S$ that contains $v_0$ is not
larger than the weight of the Steiner tree on the same nodes in $G$.
\end{enumerate}
Note that this embedding is different from the usual embedding since
we do not refer specifically to distances between vertices.
Typically, embedding is used to prove an upper bound by simplifying
the metric space. In contrast, our embedding is used to prove a
lower bound.

In order to prove the lower bound we first establish it for a star
metric,  and then extend it for a general metric space. Note that a
lower bound on a sub-graph is not a lower bound on the ambient
graph. For example, a lower bound for MST of a metric space $G$ is
not a lower bound for $G$ since the algorithm may use the additional
edges to reduce the transition time.

\section{The Model} \label{sec:TheModel}

We formally model the TSP with time windows and service time problem
as follows. Let $G = (V,w)$ be a given metric space where $V$ is a
set of $n$ nodes and $w$ is a distance function. We are given an
online sequence of requests for service. Each request is
characterized by a pair $([r_{i}, d_{i}], v_{i})$, where $r_{i} \in
N_{+}$ and $d_{i} \in N_{+}$ are the respective arrival time and
deadline time of the request, and $v_{i} \in V$ is a node in the
metric space $G$. The time to traverse from node $v_i$ to node $v_j$
is $w(v_i, v_j)$. Serving a request at some node requires unit size
service time. The goal is to serve as many requests as possible
within their time windows $[r_i,d_i]$.

When all $r_i$ are equal to $0$ and all $d_i$ are equal to $B$ and
the service time is negligible the problem is reduced to the
well-known orienteering problem with budget $B$ and a prize for each
node which is equal to the number of requests in this node. That is,
finding a path of total distance at most $B$ that maximizes the
prize of all visited nodes.


Let $\alg(\sigma)$ ($\opt(\sigma)$) denote the throughput of the
online (respectively, optimal) algorithm with respect to a sequence
$\sigma$. We consider the benefit problem and hence $\inf_{\sigma}
\!{\rm OPT}(\sigma)/{\rm ALG}(\sigma) \geq 1$.

\section{Embedding of Metric Spaces} \label{sec:Embeddings}
In this section we describe an embedding of a general metric space
into a star metric. We begin by introducing some new definitions:

\begin{itemize}
\item We define $w({\rm T}) = \sum\limits_{e \in V} w(e)$ for a tree ${\rm T} = (V,E)$, and
let $P_{{\rm T}}(v)$ denote the parent of node $v$ in a rooted tree
${\rm T}$.

\item Let ${\rm S}$ be a star metric with a center $c$. Let $w_i$ denote the weight of the edge incident to
the vertex $v_i$. We define $w_{\rm S}(V) = \sum\limits_{v \in V} w(c, v)$ ($= \sum\limits_{v_i \in V} w_i$).
It is clear that for a star ${\rm S}$ with leaves $V$, $w_{\rm S}(V) = w({\rm S})$.

\item Let ${\rm T}_{\rm G}(V)$ be the minimum weight connected component that contains the set $V$
(i.e., the minimum Steiner tree on these points) in the metric space
$G$.

\end{itemize}
Note that $E({\rm G})$ is the set of edges of graph $G$. Recall that
MST(G) denote the weight of the minimal spanning tree (MST) in the
metric space G.
\begin{theorem} \label{thm:GraphToStar}
For any given metric space ${\rm G}$ on nodes $V$ and for any vertex
$v_0 \in V$ there exists a star metric ${\rm S}$ with leaves $V$ and
an embedding $f: {\rm G} \rightarrow {\rm S}$ from  ${\rm G}$  to
${\rm S}$ $(f$ depends on $v_0)$ such that:
\begin{enumerate}
\item $w({\rm S}) = w({\rm T}_{\rm G}(V))$  {\rm (= MST(G))}.

\item For every $V' \subseteq V$ such that $v_0 \in V'$,
$w({\rm T}_{\rm G}(V')) \geq w_{\rm S}(V').$
\end{enumerate}
\end{theorem}

\begin{proof}
We prove the theorem by describing a star metric with the required
properties. Let $G$ be a given metric space on nodes $V$ and a leaf
$v_0 \in V$. Let $T$ be the MST for $G$ created by means of Prims
algorithm with the root $v_0$. Let  $S$  be a star metric with
leaves $V$ such that for each $u \in V$, $w_u = w(u, P_{{\rm
T}}(u))$. Clearly, $w_{v_0} = 0$. We prove that $S$ and $v_0$
satisfy the theorem's properties: {\\\bf Property 1:} Clearly,
$w({\rm S}) = w({\rm T})$, and since $T$ is a MST for $G$, $w({\rm
S}) = w({\rm T}) = {\rm MST(G)} = w({\rm T}_{\rm G}(V))$. {\\\bf
Property 2:} We have to prove that for every $V' \subseteq V$ such
that $v_0 \in V'$, $w({\rm T}_{\rm G}(V')) \geq w_S(V').$ Let $V' =
\{v_0, v_{i_1}, ..., v_{i_r-1}\}$. Recall that we defined $w_u =
w(u, P_{\rm T}(u))$. Clearly,  $w_{\rm S}(V') = \sum_{j=1}^{r-1}
w(v_{i_j}, P_{{\rm T}}(v_{i_j}))$. Hence it  suffices to prove that
$w({\rm T}_{\rm G}(V')) \geq \sum_{j=1}^{r-1} w(v_{i_j}, P_{\rm
T}(v_{i_j}))$. The proof is based on the following idea. We begin
with the minimum Steiner tree that contains $V'$ (meaning ${\rm
T}_{\rm G}(V')$). Then we transform it to an MST on all vertices by
running Prim from $v_0$ and replacing the Steiner tree's edges with
Prim's edges. We prove that each time the algorithm adds an edge $e$
that corresponds to an edge in $w_S(V')$ it deletes an edge $e'$
from ${\rm T}_{\rm G}(V')$ such that $w(e) \leq w(e')$. Note that we
also add edges incident to vertices not in $V'$ in order to maintain
a tree. The weights of these edges are ignored. Since the algorithm
starts with ${\rm T}_{\rm G}(V')$ and finishes with T, this proves
that the property holds (recall that the weight of the edges of  S
is determined by the weight of the edges of  T). The exact
description of our algorithm, called the Embed-Prim algorithm, is
provided in Figure \ref{alg:EmbedPrim}.

\begin{figure}[htbp]
\MyFrame{

\begin{enumerate}
\item ${\rm T}' \leftarrow {\rm T}_{\rm G}(V')$ \label{alg:Init}

\item $i \leftarrow 1$ \label{alg:AfterInit}

\item $V_{\rm new} = {v_0} ($let $u_0 = v_0)$

\item Repeat until $V_{\rm new} = V$

\begin{enumerate}

\item Choose an edge $e_i = (w,u_i)$ with minimal weight such that $w$ is in
$V_{\rm new}$ and $u_i$ is not (ties are broken by id).

\item Add $u_i$ to $V_{\rm new}$

\item If $u_i \notin V'$ then add $e_i$ and $u_i$ to ${\rm T}'$. \label{alg:Case1}
\begin{enumerate}
\item Else, if $e_i \in E({\rm T}')$ do nothing (for the proof
view $e_i$ as an edge that was deleted and added). \label{alg:Case2}

\item Else, if $e_i \notin E({\rm T}')$ then \label{alg:Case3}
\begin{enumerate}
\item Add $e_i$ to $T'$.

\item Let $C'$ be the cycle created by adding $e_i$ to ${\rm T}'$. Let $e'$ be the edge with
the maximal weight on $C'$ such that $e' \notin \{e_1, ..., e_{i}\}$
and $e' \cap \{u_0, ..., u_{i-1}\} \neq \emptyset$ (in Lemma
\ref{lemma:EmbedPrimExistEdge} we prove that such an edge always
exists).
\\ Remove $e'$ from ${\rm T}'$.
\end{enumerate}
\end{enumerate}

\item $i \leftarrow i + 1$

\end{enumerate}
\end{enumerate}

} \caption{Algorithm Embed-Prim.} \label{alg:EmbedPrim}
\end{figure}

First we show the correctness of Embed-Prim:
\begin{lemma} \label{lemma:EmbedPrimExistEdge}
Let $C'$ be the cycle created in step \ref{alg:Case3} in the
algorithm. There exists at least one edge $e'$ that belongs to $C'$,
such that $e' \notin \{e_1, ..., e_{i}\}$ and $e' \cap \{u_0, ...,
u_{i-1}\} \neq \emptyset$.
\end{lemma}
\begin{proof}
Note that a cycle is created in step \ref{alg:Case3} in the
algorithm since adding an edge to a tree always creates a cycle.
Similar to Prim, the edges that Embed-Prim adds after step
\ref{alg:Init} in the algorithm do not create a cycle. Therefore
$C'$ must contains edges added in step \ref{alg:Init} in the
algorithm. At least one of these edges must touch one of the
vertices $\{u_0, ..., u_{i-1}\}$
\end{proof}

\begin{lemma} \label{lemma:EmbedPrimTree}
After each step of Embed-Prim, ${\rm T}'$ is a tree which contains
$V'$.
\end{lemma}

\begin{proof}
At the beginning ${\rm T}'$ is ${\rm T}_{\rm G}(V')$, which is a
tree that contains $V'$. We never remove vertices and hence ${\rm
T}'$ always contains $V'$. Whenever we add an edge that creates a
cycle we open the cycle by removing an edge from it.
\end{proof}

Now we claim that Embed-Prim satisfies the following invariant:
\begin{lemma} \label{lemma:EmbedPrimInvariant}
Each time Embed-Prim adds an edge $e$ that corresponds to an edge in
$w_{\rm S}(V')$, it deletes an edge $e'$ from ${\rm T}_{\rm G}(V')$
such that $w(e) \leq w(e')$.
\end{lemma}
\begin{proof}
Step \ref{alg:Case1} in the algorithm is irrelevant, since the edge
does not correspond to an edge in $w_{\rm S}(V')$ (the vertex that
was added by Embed-Prim is not in $V'$). In step \ref{alg:Case2} in
the algorithm, $w(e) = w(e')$. In step \ref{alg:Case3} in the
algorithm, since Embed-Prim could have added edge $e$', but did
choose the edge $e$ instead, $w(e) \leq w(e')$ (recall that
Embed-Prim always chooses the edge with the minimal weight).
\end{proof}

\par Now we are ready to prove that $S$ satisfies the second property of the embedding.
By the definition of Prim $e_i = (u_i, P_{{\rm T}}(u_i))$. Hence,
$\sum_{j=1}^{r-1} w(e_{i_j}) \!=\! \sum_{j=1}^{r-1} w(v_{i_j},
P_{{\rm T}}(v_{i_j}))$. Let $e'_i$ be the edge deleted from $T'$
when edge $e_i$ was added (steps \ref{alg:Case2}, \ref{alg:Case3} in
the algorithm). Then
$$
w_{\rm S}(V') = \sum_{j=1}^{r-1} w(v_{i_j}, P_{{\rm T}}(v_{i_j}))
=\sum_{j=1}^{r-1} w(e_{i_j}) \leq \sum_{j=1}^{r-1} w(e'_{i_j}) \leq
w({\rm T}_{\rm G}(V')).
$$
where the first equality follows from the definition, the first
inequality results from the invariance, and the last inequality
follows from the definition.
\end{proof}

\section{Lower Bounds} \label{sec:Hardness}

\subsection{Lower Bound for a Large Diameter Laxity Ratio (Case A)}
In this section we consider the case where $L < \Delta(G) / 2$
(recall that $\Delta(G)$ is the diameter and $L$ is the laxity) then
we show that the competitive ratio of any algorithm is unbounded.

\begin{theorem} \label{thm:LargeDiamRatio}
No online algorithm can achieve a bounded competitive ratio for any
metric space in which $L < \Delta(G) / 2$.
\end{theorem}

\begin{proof}
Let $G$ be any metric space. Every $\Delta(G) + 1$ units of time we
bring a request with laxity of $L$ on a node which is at a distance
of at least $\Delta(G) / 2$ from the current location of the online
algorithm. It is clear that the algorithm can not serve any requests
while ${\rm OPT}$ can serve all the requests.
\end{proof}

\subsection{Lower Bound for a Small Diameter Laxity Ratio (Case B and C)} \label{sec:SmallDiamRatio}
In this section we consider Cases B and C. Let $\delta = TSP(G) /
L$. If $\delta < 1$ (Case C) We show a strictly larger than $1$
lower bound. Specifically, if $\delta \leq \frac{1}{9}$ we provide a
lower bound of $1 + \Omega\big(\sqrt{\delta}\big)$. If $\delta >
\frac{1}{9}$ (Case B) we can use requests with laxity of $9TSP({\rm
G})$ (i.e., $\delta = \frac{1}{9}$), and obtain a lower bound of
$1.002$. Therefore, from now on we only consider Case C.

\subsubsection{Lower Bound for a Star Metric}
\label{sec:LowerBoundStar}

In this section we consider the case where the traveling time
between nodes is represented by a star metric. This is also
equivalent to the case where the traveling time from node $i$ is
$w_{i}$.

The general idea is that the adversary creates many requests with
large deadline at node $v_0$ at each time unit, and also blocks of
fewer requests with close deadlines at other nodes. Any online
algorithm must choose between serving many requests with large
deadline or traveling between many of the nodes and serving also the
requests with close deadline.

Recall that $w({\rm S})$ denote the weight of the star metric $S$
(i.e., the sum of the weights of the edges of $S$). We define $F =
\sqrt{w({\rm S})L}$. Let $\delta = {\rm TSP(G)} / L = w({\rm S}) /
L$.
\begin{theorem} \label{thm:LowerBoundStar}
No deterministic or randomized online algorithm can achieve a
competitive ratio better than $1 + \Omega\big(\sqrt{\delta}\big)$ in
any given star metric $\rm S$ for $\delta \leq \frac{1}{9}$.
Otherwise, if $\delta > \frac{1}{9}$, the bound becomes $1.002$.

\end{theorem}

\begin{proof}
Let $S$ be a given star metric with nodes $V = {v_0,...,v_{n-1}}$.
We can assume, without loss of generality, that $\delta \leq
\frac{1}{9}$, since otherwise one may use requests with laxity of
$9w({\rm S})$ (i.e., $\delta = \frac{1}{9}$), and obtain a lower
bound of $1.002$. Let type A node denote node $v_0 \in V$ and type B
node denote nodes $v_1,...,v_{n-1} \in V$. Let type A request and
type B request refer to requests with type A node and type B node,
respectively. Recall that $w_i$ denote the weight of the edge incident to
the vertex $v_i$.
\\ We begin by describing the sequence $\sigma({\rm S}, \alg)$.
{\\\bf Sequence structure:} Recall that each request is
characterized by a pair $([r_{i}, d_{i}], v_{i})$, where $r_{i} \in
N_{+}$ and $d_{i} \in N_{+}$ are the respective arrival time and
deadline time of the request, and $v_{i}$ is a node in $S$. There
are up to $N=\frac{L}{3F}=\frac{1}{3}\sqrt{\frac{L}{w({\rm S})}}$
blocks, where each block consists of $3F$ time units. Let $t_i = 1 +
3(i-1)F$ denote the beginning time of block $i$. For each block $i$,
where $1 \leq i \leq N$, $F$ requests located at various nodes
arrive at the beginning of the block. Specifically,
$\frac{w_j}{w({\rm S})-w_0}F$ type B requests $([t_i,L+t_i],v_j)$,
for each $1 \leq j \leq n-1$ are released. A type A request
$([t,3L],v_0)$ is released at each time unit $t$ in each block. Once
the adversary stops the blocks, additional requests arrive (we call
this the final event). The exact sequence is defined as follows:
\begin{enumerate}
\item $i \leftarrow 1$
\item Add block $i$ \label{subsec:AddBlockStar}
\item If with probability at least 1/4
there are at least $F/2$ unserved type B requests at the end of block $i$
(denoted by Condition 1),
then $L$ requests $([t_{i+1},L+t_{i+1}],v_1)$ are released and the sequence is terminated.
See Figure \ref{figure:Case1}.
Clearly, $t_{i+1}$ is the time of the final event. Denote this by Termination Case 1.

\item Else, if with probability at least 1/4,
at most $2F$ requests were served during block $i$ (denoted by
Condition 2), then $3L$ requests $([t_{i+1},3L],v_0)$ are released
and the sequence is terminated. Clearly, $t_{i+1}$ is the time of
the final event. See Figure \ref{figure:Case2}. Denote this by
Termination Case 2.

\item Else, if $i = N$ (there are $N$ blocks, none of them satisfied Condition 1 or 2)
$3L$ requests $([L+1,3L],v_0)$ are released, and the sequence is
terminated. Clearly, $L+1$ is the time of the final event. See
Figure \ref{figure:Case3}. Denote this by Termination Case 3.

\item Else ($i < N$) then $i \leftarrow i + 1$, Goto \ref{subsec:AddBlockStar}

\end{enumerate}
We make the following {\bf observations:} (i) Each block consists of
$3F$ time units. Hence, if $\alg$ served at most $2F$ requests
during a block, there must have been at least $F$ idle time units.
(ii) There are up to $\frac{1}{3}\sqrt{\frac{L}{w({\rm S})}}$ blocks
and each block consists of 3$\sqrt{w({\rm S})L}$ time units. Hence,
the time of the final event is at most $L+1$. (iii) Exactly one type
A request arrives at each time-slot until the final event. Hence, at
most $L$ type A requests arrive before (not including) the final
event. (iv) During each block, exactly $F$ type B requests arrive,
which sum up to at most $L/3$ type B requests before (not including)
the final event.

Now we can analyze the competitive ratio of $\sigma({\rm S}, \alg)$.
Consider the following possible sequences (according to the termination type):
\begin{enumerate}
\item Termination Case 1: Let $Y$ denote the number of requests in the sequence. According to the observations, the sequence
consists of at most $L$ type A requests, and at most $\frac{4}{3}L$ type B requests ($L/3$ until the
final event and $L$ at the final event). Hence, $Y \leq L + \frac{4}{3}L \leq 3L$.

\begin{itemize}
\item {\bf We bound the performance of ALG:}
At time $t_{i+1}$ there is a probability of at least 1/4 that $\alg$
has $L + F/2$ unserved type B requests. Since type B requests have
laxity of $L$, $\alg$ can serves at most $L+1$ of them, and must
drop at least $F/2 - 1$. The expected number of served requests is
$$
E(\alg(\sigma)) \leq Y - \frac{1}{4}(F/2 - 1) = Y - \frac{1}{8}F + 1/4.
$$
\item {\bf We bound the performance of OPT$^\prime$:}
$\opt ^\prime$ serves the requests in three stages:
\begin{itemize}
\item {\bf Type B requests that arrive before the final event:}
Recall that all type B requests in a block arrive at once in the
beginning of the block. In each block $\opt ^\prime$ serves first
all the requests to node $v_{1}$, then all the requests to node
$v_{2}$, and so on. It is clear that $\opt ^\prime$ needs at most
$F+2w({\rm S})$ time units to serves the requests ($F$ for serving
and $2w({\rm S})$ for traveling). $\opt ^\prime$ serves the requests
starting from the beginning of the block. Recall that $L \geq
9w({\rm S})$ and $F = \sqrt{w({\rm S})L}$. Therefore $2F \geq
18w({\rm S})$. Since the block's size is $3F$, there are enough time
units. Moreover, since $L \geq 9w({\rm S})$, $L \geq 3\sqrt{w({\rm
S})L} = 3F > F+2w({\rm S})$. Hence, all the requests can be served
before deadline expiration.
\item {\bf Type B requests that arrive during the final event:}
The $L$ requests $([t_{i+1},$ $L+t_{i+1}],v_1)$ arrived during the
final release time are served by $\opt ^\prime$ consecutively from
time $t_{i+1}$. $\opt ^\prime$ can serve $L$ requests, except for
one travel phase, and hence may lose at most $2w({\rm S})$ requests.
According to the observations, the time of the final event $t_{i+1}$
is at most $L+1$. We conclude that $\opt ^\prime$ serves all type B
requests until time unit $2L$.

\item {\bf Type A requests:}
$\opt ^\prime$ serves the $L$ type A requests consecutively from
time unit $2L+1$. Since the deadlines are $3L$, $\opt ^\prime$
serves all type A requests.
\end{itemize}
We conclude that $\opt(\sigma) \geq \opt ^\prime (\sigma) \geq Y - 2w({\rm S})$.
\end{itemize}
The competitive ratio is
\begin{eqnarray*}
\frac{\opt(\sigma)}{E(\alg(\sigma))} & \geq & \frac{Y - 2w({\rm
S})}{Y - \frac{1}{8}F + 1/4} \geq \frac{3L - 2w({\rm S})}{3L -
\frac{1}{8}\left(\sqrt{w({\rm S})L}\right) + 1/4} =  1 +
\Omega\left(\sqrt{\frac{w({\rm S})}{L}}\right) \ .
\end{eqnarray*}

Here the second inequality results from the fact that the number is
above 1 and the numerator and the denominator increase by the same
value.

\item Termination Case 2: The sequence consists of more than $3L$ type A requests, all deadlines are
at most $3L$.

\begin{itemize}
\item {\bf We bound the performance of ALG:}
The probability that $\alg$ was idle during $F$ time units is at
least 1/4. Hence, the expected number of served requests is
$E(\alg(\sigma)) \leq 3L - \frac{1}{4}F. $

\item {\bf We bound the performance of OPT$^\prime$:}
At each time unit until the final event, $\opt ^\prime$ serves the
type A request that arrived at that particular time unit.
Consequently, from the final event and until time unit $3L$, $\opt
^\prime$ serves the type A requests that arrived at the final event.
Therefore, $\opt ^\prime$ serves $3L$ type A requests, and so
$\opt(\sigma) \geq \opt ^\prime (\sigma) \geq 3L$.
\end{itemize}

The competitive ratio is
$$
\frac{\opt(\sigma)}{E(\alg(\sigma))} \geq \frac{3L}{3L -
\frac{1}{4}F} = \frac{3L}{3L - \frac{1}{4}\left(\sqrt{w({\rm
S})L}\right)} = 1 + \Omega\left(\sqrt{\frac{w({\rm S})}{L}}\right).
$$

\item Termination Case 3: the sequence consists of $3L$ type A requests, all deadlines are
at most $3L$.

\begin{itemize}
\item {\bf We bound the performance of ALG:}
Let $U_{i}$ be the event that the number of unserved type B requests
at the end of block $i$ is less than $F/2$. If $U_{i}$ occurs, then
let $j_{k}$, $1 \leq k \leq r$, be the type B nodes visited by
$\alg$ in block $i$. At least $F/2$ requests that arrived in this
block have to be served (recall that $F$ type B requests arrive at
the beginning of each block). Therefore,
$$
\frac{w_{j_{1}}}{w({\rm S})-w_0}F + \frac{w_{j_{2}}}{w({\rm S})-w_0}F + \cdots + \frac{w_{j_{r}}}{w({\rm S})-w_{0}}F \geq F/2,
$$
and so
$$
w_{j_{1}} + w_{j_{2}} + \cdots + w_{j_{r}} \geq \frac{w({\rm S})-w_{0}}{2}.
$$
Let $E_{i}$ be the event that more than $2F$ requests are served
during block $i$. If event $U_{i-1}$ and $E_{i}$ occur, then there
are at most $3F/2$ unserved type B requests in the beginning of
block $i$ ($F$ arrived in the beginning of the block and at most
$F/2$ from the previous block) but more than $2F$ requests were
served. Therefore, at least one type A request was served during the
block. Combining the results, if $U_{i}$, $U_{i-1}$ and $E_{i}$
occur then:
\begin{itemize}
\item During block $i$ at least $(w({\rm S})-w_{0})/2$ time units were used for traveling
between type B nodes.

\item Type A request was served during the block.
\end{itemize}
A block $i$ is called {\it good\/} if the events $U_{i}$, $U_{i-1}$
and $E_{i}$ occur. For any two (consecutive) good blocks the
traveling cost is at least $(w(S)-w_{0})/2  + w_{0} \geq w(S)/2$.
Since none of the blocks satisfy Condition 1 or 2, it follows that
for all $i$ such that  $\frac{1}{3}\sqrt{\frac{L}{w({\rm S})}} \geq
i \geq 1$ we have: ${\rm Pr}[U_{i}] \geq 3/4, {\rm Pr}[U_{i-1}] \geq
3/4,$ and ${\rm Pr}[E_{i}] \geq 3/4$. Therefore:
\begin{align*}
& {\rm Pr}[U_{i} \cap U_{i-1} \cap E_{i}]   =   1 - {\rm Pr}[\neg (U_{i} \cap U_{i-1} \cap E_{i})]  \\
& \qquad =   1 - {\rm Pr}[\neg U_{i} \cup \neg U_{i-1} \cup \neg E_{i}] \geq 1 - 1/4 - 1/4 - 1/4 = 1/4.
\end{align*}
The sequence consists of $\frac{1}{3}\sqrt{\frac{L}{w({\rm S})}}$
blocks. Therefore, the expected number of good blocks is
$\frac{1}{4} \cdot \frac{1}{3}\sqrt{\frac{L}{w({\rm S})}} =
\frac{1}{12}\sqrt{\frac{L}{w({\rm S})}}$ and hence the expected
number of disjoint pairs of blocks is
$\frac{1}{24}\sqrt{\frac{L}{w({\rm S})}}$. Consequently, the
expected number of lost requests is at least
$\frac{1}{24}\sqrt{\frac{L}{w({\rm S})}} \frac{w({\rm S})}{2}$. We
conclude that the expected number of served requests is
$$
E(\alg(\sigma)) \leq 3L - \frac{1}{48}w({\rm S})\sqrt{\frac{L}{w(S)}}
= 3L - \frac{1}{48}\left(\sqrt{w({\rm S})L}\right).
$$

\item {\bf We bound the performance of OPT$^\prime$:}
At each time unit until the final event, $\opt ^\prime$ serves the
type A request that arrived at the same time unit. Consequently,
from the final event and until time unit $3L$, $\opt ^\prime$ serves
the type A requests that arrived at the final event. Therefore,
$\opt ^\prime$ serves 3L type A requests, and so
 $\opt \geq \opt ^\prime \geq 3L$.

\end{itemize}
The competitive ratio is
$$
\frac{\opt(\sigma)}{E(\alg(\sigma))} \geq \frac{3L}{3L -
\frac{1}{48}\left(\sqrt{w({\rm S})L}\right)} = 1 +
\Omega\left(\sqrt{\frac{w({\rm S})}{L}}\right).
$$

Note that in all 3 cases we got $1 + \Omega\left(\sqrt{\frac{w({\rm
S})}{L}}\right) = 1 + \Omega\big(\sqrt{\delta}\big)$. This completes
the proof.
\end{enumerate}
\end{proof}

\begin{corollary} \label{thm:LowerBoundUniform}
No deterministic or randomized online algorithm can achieve a
competitive ratio better than $1 + \Omega\big(\sqrt{n / L}\big)$
when all traveling times takes one unit of time and $L \geq 9n$.
Otherwise, if $L < 9n$, the bound becomes $1.002$ (note that in this
case $\delta = n / L$).
\end{corollary}

\begin{proof}
Let $S$ be a star metric with $n$ nodes such that the weight of each
edge is equal to 1/2. Clearly, traveling between each two nodes
requires one time unit and $w({\rm S}) = n/2$. Applying Theorem
\ref{thm:LowerBoundStar}, we obtain the lower bound of $1 +
\Omega\big(\sqrt{n / L}\big)$.
\end{proof}

Observe that if $n = C$ the bound becomes $1 + \Omega\big(\sqrt{C /
L}\big)$ which improves the lower bound of $1 + \Omega\big(C /
L\big)$ from \cite{Azar_BMC} (since $C < L$).


\subsubsection{Lower Bound for a General Metric Space} \label{sec:LowerBoundGraph}

In this section we consider the case where the traveling time
between nodes is represented by a metric space $G$. Note that a
lower bound for a star metric space does not imply a lower bound for
a general metric space. Recall that $\delta = {\rm TSP(G) }/ L < 1$.

We use the embedding from Theorem \ref{thm:GraphToStar} to prove a
$1 + \Omega\big(\sqrt{\delta}\big)$ lower bound.
\begin{theorem} \label{thm:LowerBoundMetricSpace}
No deterministic or randomized online algorithm can achieve a
competitive ratio better than $1 + \Omega\big(\sqrt{\delta}\big)$ in
any given metric space ${\rm G}$, for $\delta \leq \frac{1}{9}$.
Otherwise, if $\delta > \frac{1}{9}$, the bound becomes $1.002$.
\end{theorem}

\section{Upper bounds} \label{sec:Upper}

\subsection{Constant Approximation Algorithm for case B} \label{sec:AlgorithmB}
In this section we design a deterministic online algorithm, for a
general metric space where $L > 9\Delta(G)$ (recall that $\Delta(G)$
is the diameter of $G$). The algorithm achieves a constant
competitive ratio. As shown in the previous section, no online
algorithm can achieves a competitive ratio better $1.002$. A more
precise analysis can replace $L > 9\Delta(G)$ with $L > (2 +
\epsilon)\Delta(G)$ for any $\epsilon > 0$.

The algorithm which we call ORIENT-WINDOW combines the following
ideas.
\begin{itemize}
\item The algorithm works in phase of $K = 3\Delta(G)$.
In each phase the algorithm serves only requests that arrived in the
previous phases, and will not expired during the phase. Due to this
perturbation we lose a constant factor.

\item The decision which requests will be served in a phase ignore
their deadlines. Due to this violation of EDF we lose a constant
factor.

\item In each phase the algorithm serves requests node by node.
The order of the nodes is determined by solving an orienteering
problem. Since a constant approximation algorithm is known to the
orienteering problem, we lose a constant factor.
\end{itemize}

\vspace{-10pt}
\begin{figure}[htbp]
\MyFrame{
\smallskip In each phase $\ell = 1, 2, \ldots$,\  do
\begin{itemize}
\item Beginning of the phase (at time $K(\ell-1)$)
\begin{itemize}
\item Decrease the deadline of each unserved request $(r,d,v)$
from $d$ to $K \floor{d/K}$.

\item Let $R^\ell$ be the collection of unserved requests
such that their decreased deadline was not exceeded. Let $S_j^{\ell}
\subseteq R^{\ell}$ denote the subset of requests at node $v_j$ in
$R^{\ell}$.
\item Using a constant approximation algorithm solve the
unrooted orienteering problem with budget $\Delta(G)$ where the
prize of a node $v_j$ is the number of requests in $S_j^{\ell}$. Let
$v_{i_{1}}$, $v_{i_{2}}$, ..., $v_{i_{r}}$ denote the order of the
nodes in the solution.

\item $\rho_\ell$ consists of all requests of $S_{i_{1}}^{\ell}$
scheduled consecutively, then all requests of $S_{i_{2}}^{\ell}$
scheduled consecutively, and so on.

\end{itemize}

\item During the phase (between time $K(\ell-1)$ and time $K\ell$)

\begin{itemize}

\item The requests are served according to $\rho_\ell$ (unserved requests in the suffix of
$\rho_\ell$, due to the end of the phase are dropped).

\end{itemize}

\end{itemize}

} \caption{Algorithm ORIENT-WINDOW} \label{alg:TSPEDFB}
\vspace{-10pt}
\end{figure}

\begin{theorem} \label{thm:ApproximationB}
The algorithm ORIENT-WINDOW attains a competitive ratio of $O(1)$.
\end{theorem}

\subsection{Asymptotically Optimal Algorithm for Case C} \label{sec:Algorithm}
In this section we design a deterministic online algorithm, for a
general metric space. The algorithm achieves a competitive ratio of
$1 + o(1)$ when the minimum laxity of the requests is asymptotically
larger than the weight of the TSP (as shown in the previous
sections, this is essential).

The algorithm is a natural extension of the BG algorithm from
\cite{Azar_BMC}. Our algorithm, which we call TSP-EDF, formally
described in Figure \ref{alg:TSPEDF}, works in phases of $K =
\sqrt{{\rm TSP(G)}L}$ time units. In each phase the algorithm serves
requests node by node. The order of the nodes is determined by the
minimum TSP or an approximation. The algorithm achieves a
competitive ratio of $1 + O\big(\sqrt{{\rm TSP(G)}/ L}\big)$ for $L
> {\rm 10TSP(G)}$.

\vspace{-10pt}
\begin{figure}[htbp]
\MyFrame{
\smallskip In each phase $\ell = 1, 2, \ldots$,\  do
\begin{itemize}
\item Beginning of the phase (at time $K(\ell-1)$)
\begin{itemize}
\item Decrease the deadline of each unserved request $(r,d,v)$
from $d$ to $K \floor{d/K}$.

\item Let $R^\ell$ be the collection of unserved requests
such that their decreased deadline was not exceeded. Let $S^{\ell}$
be the $K$-length prefix of EDF (earliest deadline first) schedule
(according to the modified deadline) of $R^\ell$. Let $S_j^{\ell}
\subseteq S^{\ell}$ denote the subset of requests at node $v_j$ in
$S^{\ell}$.
\\ Let $v_{i_{1}}$, $v_{i_{2}}$, ..., $v_{i_{n}}$ denote the order of the
nodes in the minimal TSP (or approximation).

\item $\rho_\ell$ consists of all requests of $S_{i_{1}}^{\ell}$
served consecutively, then all requests of $S_{i_{2}}^{\ell}$ served
consecutively, and so on.

\end{itemize}

\item During the phase (between time $K(\ell-1)$ and time $K\ell$)

\begin{itemize}

\item The requests are served according to $\rho_\ell$ (unserved requests in the suffix of
$\rho_\ell$, due to the end of the phase are dropped).

\end{itemize}

\end{itemize}

} \caption{Algorithm TSP-EDF.} \label{alg:TSPEDF} \vspace{-10pt}
\end{figure}

\begin{theorem} \label{thm:Approximation}
The algorithm {\rm TSP-EDF} attains a competitive ratio of $1 +
O\big(\sqrt{{\rm TSP(G)} / L}\big)$.
\end{theorem}


\bibliographystyle{abbrv}
\bibliography{SchedulingColoredPackets}

\appendix
\section{Proofs} \label{appsec:Proofs}

\subsection{Proof of Theorem \ref{thm:LowerBoundMetricSpace}} \label{appsec:MetricSpaceProof}
Let  $G$  be a given metric space on nodes $V$. We use the embedding
from Theorem \ref{thm:GraphToStar}. Let ${\rm S}$, $v_0$ be the
output of the embedding. Let $\sigma({\rm S}, \alg)$ be the sequence
described in Theorem \ref{thm:LowerBoundStar}, when $v_0$ is type A
node and the other nodes are type B. Recall that, by definition, $F
= \sqrt{w({\rm S})L}$. We use $\sigma$ for $\alg$ on $G$. We can
assume, without loss of generality, that $\delta \leq \frac{1}{9}$
since otherwise one may use requests with laxity of $9TSP(G)$ (i.e.,
$\delta = \frac{1}{9}$), and obtain a lower bound of $1.002$.
Note that for any metric space $G$ we have ${\rm MST(G)} \leq {\rm TSP(G)} < 2{\rm MST}(G)$.
Consider the following possible cases, similar to the proof of
Theorem \ref{thm:LowerBoundStar}.

\begin{enumerate}
\item In Termination Case 1 there exists a block $i$ such that, with probability at least 1/4, at the
end of the block there are at least $F/2$ unserved type B requests.
In Theorem \ref{thm:LowerBoundStar} we proved that:
\begin{itemize}
\item The sequence consists of up to $3L$ requests.

\item The expected number of requests $\alg$ missed is at least $F/8 - 1/4$.

\item $\opt$ missed up to ${\rm TSP(G)}$ requests (while transmitting the type B requests
that arrived during the final event).
\end{itemize}
Therefore, the competitive ratio depends only on $F$, TSP(G) and $L$:
\begin{eqnarray*}
\frac{\opt(\sigma)}{E(\alg(\sigma))} &\geq& \frac{3L - {\rm TSP(G)}}{3L - \frac{1}{8}F + 1/4} = \frac{3L - {\rm TSP(G)}}{3L - \frac{1}{8}\left(\sqrt{w({\rm S})L}\right) + 1/4} \\
&=& \frac{3L - {\rm TSP(G)}}{3L - \frac{1}{8}\left(\sqrt{{\rm
MST(G)}L}\right) + 1/4} = 1 + \Omega\left(\sqrt{{\rm TSP(G) }/
L}\right) \ .
\end{eqnarray*}
Here the second equality results from the fact that $w({\rm S}) = {\rm MST(G)}$.

\item In Termination Case 2 there exists a block $i$ such that, with probability at least 1/4, at most $2F$ requests
were served during the block. In Theorem \ref{thm:LowerBoundStar} we
proved that:
\begin{itemize}
\item At most $3L$ requests can be served.

\item The expected number of requests $\alg$ missed is at least $F/4$.

\item $\opt ^\prime$ served 3L type A requests.
\end{itemize}
Therefore, $\opt \geq \opt ^\prime = 3L$ and the competitive ratio depends only on $F$ and $L$:
\begin{eqnarray*}
\frac{\opt(\sigma)}{E(\alg(\sigma))} &\geq& \frac{3L}{3L - \frac{1}{4}F} = \frac{3L}{3L - \frac{1}{4}\left(\sqrt{w({\rm S})L}\right)} \\
&=& \frac{3L}{3L - \frac{1}{4}\left(\sqrt{{\rm MST(G)}L}\right)}
= 1 + \Omega\left(\sqrt{{\rm TSP(G)} / L}\right)\ .
\end{eqnarray*}
Here the second equality results from the fact that $w({\rm S}) = {\rm MST(G)}$.

\item In Termination Case 3 $\alg$ served type A request and at least $F/2$ type B requests at
each block. In Theorem \ref{thm:LowerBoundStar} we proved that:
\begin{itemize}
\item At most $3L$ requests can be served.

\item The expected number of requests $\alg$ missed in each block
due to traveling time is at least $\frac{1}{8} \frac{w({\rm S})}{2}
= \frac{w({\rm S})}{16}$.

\item $\opt ^\prime$ served $3L$ type A requests. Therefore $\opt \geq \opt ^\prime = 3L$.
\end{itemize}
By the first property required by this theorem, each sequence of
nodes in $G$ requires more traveling time than in S. Therefore, the
expected number of requests $\alg$ missed per block is at least
$\frac{w({\rm S})}{16}$. Since the number of blocks is
$\frac{1}{3}\sqrt{\frac{L}{w({\rm S})}}$, we conclude that the
competitive ratio is:
\begin{align*}
\frac{\opt(\sigma)}{E(\alg(\sigma))} &\geq  \frac{3L}{3L - \left(\frac{1}{3}\sqrt{L/w(S)} \right)\frac{w({\rm S})}{16}}\\
&= \frac{3L}{3L - \left(\frac{1}{3}\sqrt{L/{\rm MST(G)}} \right)\frac{{\rm MST(G)}}{16}}
= 1 + \Omega\left(\sqrt{TSP(G) / L}\right)\ .
\end{align*}
Here the first equality result from the fact that $w(S) = MST(G)$.

Note that in all 3 cases we got $1 + \Omega\left(\sqrt{{\rm TSP(G)}
/ L}\right) = 1 + \Omega\big(\sqrt{\delta}\big)$.
\\ This completes the proof of the Theorem \ref{thm:LowerBoundMetricSpace}.
\end{enumerate}

\subsection{Proof of Theorem \ref{thm:ApproximationB}} \label{appsec:OnlineAlgoAnalysisB}
First we need to demonstrate that the output $\rho$ is feasible.
Specifically, we need  to prove that every request $i$ is served
during the time window $[r_{i}, d_{i}]$, and that there is a
traveling time of length $w(v_i, v_j)$ between serving any two
successive requests with different nodes $v_i$ and $v_j$.

\begin{lemma} \label{lemma:CorrectnessB}
The algorithm ORIENT-WINDOW generates a valid output.
\end{lemma}
\begin{proof}
All request in $R^\ell$ have deadlines of at least $K\ell$. Hence,
when they are served they have not been expired yet. Serving at
phase $\ell$ ends at time $K\ell$ even if not all requests in
$R^\ell$ are served. Hence, the phases of the algorithm are well
defined.
\end{proof}

Now we analyze the performance guarantee of the algorithm. Let
$\sigma'$ be the following modification of $\sigma$. It consists of
all requests in $\sigma$ such that a request $([r,d],v) \in \sigma$
is replaced by a request $([K \ceil{r/K}, K \floor{d/K}], v) \in
\sigma'$. Hence, the release and deadline times of each request in
$\sigma'$ are aligned with the start/end time of the corresponding
phase so that the time window of each request is fully contained in
the time window of that request according to $\sigma$.

The notion of \emph{$\lambda$-perturbation}, defined in
\cite{Azar_BMC}, is as follows: An input sequence $\widehat{\delta}$
is a \emph{$\lambda$-perturbation} of $\delta$ if $\widehat{\delta}$
consists of all requests of $\delta$, and each request
$[\widehat{r},\widehat{d}] \in \widehat{\delta}$ corresponding to
request $[r,d] \in \delta$ satisfies $\widehat{r} - r \leq \lambda$
and $d - \widehat{d} \leq \lambda$.
\\ By definition, $\sigma'$ is $K$-perturbation of $\sigma$.

We use Theorem $2.2$ from \cite{Azar_BMC}:
\begin{theorem}
Suppose $\widehat{\delta}$ is a $\lambda$-perturbation of $\lambda$
then $\opt(\widehat{\delta}) = (1 - 2\lambda/L)\opt(\delta)$ (where
L is the minimum laxity).
\end{theorem}

By applying the Theorem we get
\begin{equation} \label{alg:PertubB}
\opt(\sigma') \geq \big(1 -
2\frac{3\Delta(G)}{9\Delta(G)}\big)\opt(\sigma) = \opt(\sigma) / 3.
\end{equation}

Let $\opt ^\prime$ be an optimal offline algorithm that is not
allowed to serve any request that was served by the online
algorithm. It is clear that $\opt(\sigma) \leq \opt ^\prime(\sigma)
+ \alg(\sigma)$. Moreover, at each time unit the set of the unserved
requests of $\opt ^\prime(\sigma)$ is a subset of the unserved
requests of $\alg(\sigma)$.

\begin{lemma} \label{lemma:ORIENTApp}
$\alg(\sigma') \geq \opt ^\prime(\sigma') / 9$
\end{lemma}
\begin{proof}
We prove the Lemma for each phase separately. Recall that in each
specific phase $\ell$ we compute $3$-approximation to the unrooted
orienteering problem \cite{bansal2004approximation} with budget
$\Delta(G)$ where the prize of a node $v_j$ is the number of
requests in $S_j^{\ell}$ (actually there is a
$(2+\epsilon)$-approximation \cite{chekuri2012improved} but for
simplicity we use the $3$-approximation). Let $x$ be the total prize
in the solution for phase ${\ell}$. We separate into two cases. If
$x \geq \Delta(G)$ then $\alg$ serves at least $\Delta(G)$ requests.
It is clear that he needs at most $3\Delta(G)$ time units to serve
the requests (at most $\Delta(G)$ time units to travel to the first
node in the solution, at most $\Delta(G)$ time units to travel
between the nodes of the solution and $\Delta(G)$ time units to
serve the requests). Since the size of each phase is $3\Delta(G)$, there are
enough time units to serve $\Delta(G)$ requests. $\opt ^\prime$ can
serve up to $3\Delta(G)$ requests (a request per time unit). Hence,
$\alg(\sigma') \geq \opt ^\prime(\sigma') / 3$ in phase ${\ell}$.

If $x < \Delta(G)$ $\alg$ serves $x$ requests follows by a similar
argument to the one described in the previous case. Assume by a
contradiction that $\opt ^\prime$ served more than $9x$ requests.
Since the size of the phase is $3\Delta(G)$ it is clear that there exists
$\Delta(G)$ consecutive time units during the phase in which $\opt
^\prime$ served more than $9x/3 = 3x$ requests. This is a
contradiction to the correctness of the approximation algorithm
(recall that $\opt ^\prime$ is not allowed to serve requests that
were served by $\alg$. Therefore, the set of the unserved requests of
$\opt ^\prime(\sigma)$ is a subset of the unserved requests of
$\alg(\sigma)$). Hence, $\opt ^\prime$ served up to $9x$ requests
and $\alg(\sigma') \geq \opt ^\prime(\sigma') / 9$ in phase
${\ell}$.
\end{proof}

Now are ready to prove that {\rm ORIENT-WINDOW} attains a
competitive ratio of $O(1)$.

Using the previously stated results, we obtain that

$$ \alg(\sigma)  =  \alg(\sigma')
 \geq  \opt ^\prime(\sigma') / 9  \geq  \opt ^\prime(\sigma) /
(9 * 3)  =  \opt ^\prime(\sigma) / 27$$.

The first equality follows by the definition of the algorithm. The
first inequality results from Lemma~\ref{lemma:ORIENTApp} The second
inequality holds by Equation (\ref{alg:PertubB}).

Combining the result with the relation between  $\opt
^\prime(\sigma)$ and $\opt(\sigma)$ we get:

$$\opt(\sigma) \leq \opt ^\prime(\sigma) + \alg(\sigma) \leq 28\alg(\sigma)$$

This complete the proof.

\subsection{Proof of Theorem \ref{thm:Approximation}} \label{appsec:OnlineAlgoAnalysis}
First we need to demonstrate that the output $\rho$ is feasible.
Specifically, we need  to prove that every request $i$ is served
during the time window $[r_{i}, d_{i}]$, and that there is a
traveling time of length $w(v_i, v_j)$ between serving any two
successive requests with different nodes $v_i$ and $v_j$.

\begin{lemma} \label{lemma:Correctness}
The algorithm TSP-EDF generates a valid output.
\end{lemma}
\begin{proof}
The  proof is similar to the proof of Lemma
\ref{lemma:CorrectnessB}.
\end{proof}

Now we analyze the performance guarantee of the algorithm. First we
define two input sequences $\sigma'$ and $\widetilde{\sigma}$, which
are modifications of $\sigma$. The input sequence $\sigma'$ consists
of all requests in $\sigma$, but modifies the node of the requests
to a fixed node $v'$. Specifically, each request $([r,d],v) \in
\sigma$ defines a request $([r,d],v') \in \sigma'$. The input
sequence $\widetilde{\sigma}$ consists of all requests in $\sigma$
such that a request $([r,d],v) \in \sigma$ is replaced by a request
$([K \ceil{r/K}, K \floor{d/K}], v') \in \widetilde{\sigma}$, where
$v'$ is a fixed node. Hence, all requests in $\widetilde{\sigma}$
have the same node, and the release and deadline times of each
request in $\widetilde{\sigma}$ are aligned with the start/end time
of the corresponding phase so that the time window of each request
is fully contained in the time window of that request according to
$\sigma$.

\begin{lemma} \label{lemma:SigmaTildaEq}
$\opt(\widetilde{\sigma}) = \alg(\widetilde{\sigma})$.
\end{lemma}
\begin{proof}
Note that algorithm TSP-EDF has three modification with respect to EDF:
\begin{itemize}
\item request deadline times are modified to $K \floor{d/K}$.

\item request release times are modified to $K \ceil{r/K}$
(because in each phase only requests released during previous phases
are served).

\item Traveling time units are added between serving requests from different nodes.

\end{itemize}

The release and deadline times of the requests in
$\widetilde{\sigma}$ are aligned and all the requests have the same
node. Hence, $\alg$'s schedule is identical to EDF's schedule. Since
EDF is optimal for sequences that consist of requests with one node
(observe that this is the same as scheduling packets with unit size
and unit value), $\opt(\widetilde{\sigma}) =
\alg(\widetilde{\sigma})$.
\end{proof}

By definition, $\widetilde{\sigma}$ is $K$-perturbation of
$\sigma'$, and the nodes of all requests are identical. Hence,
Theorem 2.2 from \cite{Azar_BMC} (see Section
\ref{appsec:OnlineAlgoAnalysisB} for the definitions of perturbation
and Theorem 2.2) yields the following inequality:
\begin{equation} \label{alg:Pertub}
\opt(\widetilde{\sigma}) \geq \left(1 - 2\sqrt{{\rm TSP(G) }/
L}\right)\opt(\sigma').
\end{equation}

\begin{lemma} \label{lemma:TRANSPenalty}
$\alg(\sigma) \geq \big(1 - \sqrt{{\rm TSP(G)} / L}\,\big)\alg(\sigma')$.
\end{lemma}
\begin{proof}
The difference between the schedule TSP-EDF generated for $\sigma$
and the schedule it generates for $\sigma'$ is that requests might
be dropped at the end of each phase in $\sigma$ due to traveling
time. The worst case for $\sigma$ is when there are no idle time
units in any of the phases of $\sigma'$. Otherwise, the idle time
units might be used for traveling time. Therefore, there are at most
$\lceil \alg(\sigma') / K \rceil - 1$ phases in which algorithm
TSP-EDF serves requests (the $-1$ term is due to the fact that the
algorithm does not serve any request during the first phase). Since
there are at most TSP(G) traveling time units in each phase,
we obtain the following inequality:
\begin{eqnarray*}
\alg(\sigma) & \geq & \alg(\sigma') - \left(\lceil \alg(\sigma') / K \rceil - 1\right){\rm TSP(G) }\\
& = & \alg(\sigma') - \left(\left \lceil \frac{\alg(\sigma')}{\sqrt{{\rm TSP(G)}L}} \right \rceil - 1\right){\rm TSP(G)}\\
& \geq & \left(1 - \sqrt{{\rm TSP(G) }/ L}\right)\alg(\sigma') \ .
\end{eqnarray*}
\end{proof}

Now are ready to prove that {\rm TSP-EDF} attains a competitive
ratio of $1 + O\big(\sqrt{{\rm TSP(G)} / L}\big)$.

Using the previously stated results, we obtain that
\begin{eqnarray*}
\alg(\sigma) & \geq & \left( 1 - \sqrt{{\rm TSP(G) }/ L} \right) \alg(\sigma')
= \left( 1 - \sqrt{{\rm TSP(G) } / L} \right) \alg(\widetilde{\sigma}) \\
& = & \left( 1 - \sqrt{{\rm TSP(G) } / L} \right) \opt(\widetilde{\sigma}) \\
& \geq & \left( 1 - \sqrt{{\rm TSP(G) } / L} \right) \left( 1 - 2\sqrt{{\rm TSP(G) } / L} \right) \opt(\sigma') \\
& \geq & \left( 1 - 3\sqrt{{\rm TSP(G) } / L} \right) \opt(\sigma) \ .
\end{eqnarray*}
The first inequality results from Lemma~\ref{lemma:TRANSPenalty}.
The first equality follows by the definition of the algorithm. The
second equality holds by lemma~\ref{lemma:SigmaTildaEq}. The second
inequality results from Equation (\ref{alg:Pertub}). Finally, the
last inequality holds because $\sigma'$ is similar $\sigma$, but all
requests have the same node. This implies that any schedule feasible
for $\sigma$ is also feasible for $\sigma'$, and thus $\opt(\sigma')
\geq \opt(\sigma)$.~

We conclude that $\frac{\opt(\sigma)}{\alg(\sigma)} \geq 1 +
O\big(\sqrt{{\rm TSP(G)} / L}\big)$

\section{Figures} \label{appsec:Figures}
\begin{figure}[htbp]
\begin{center}
\scalebox{0.40}
{
\includegraphics{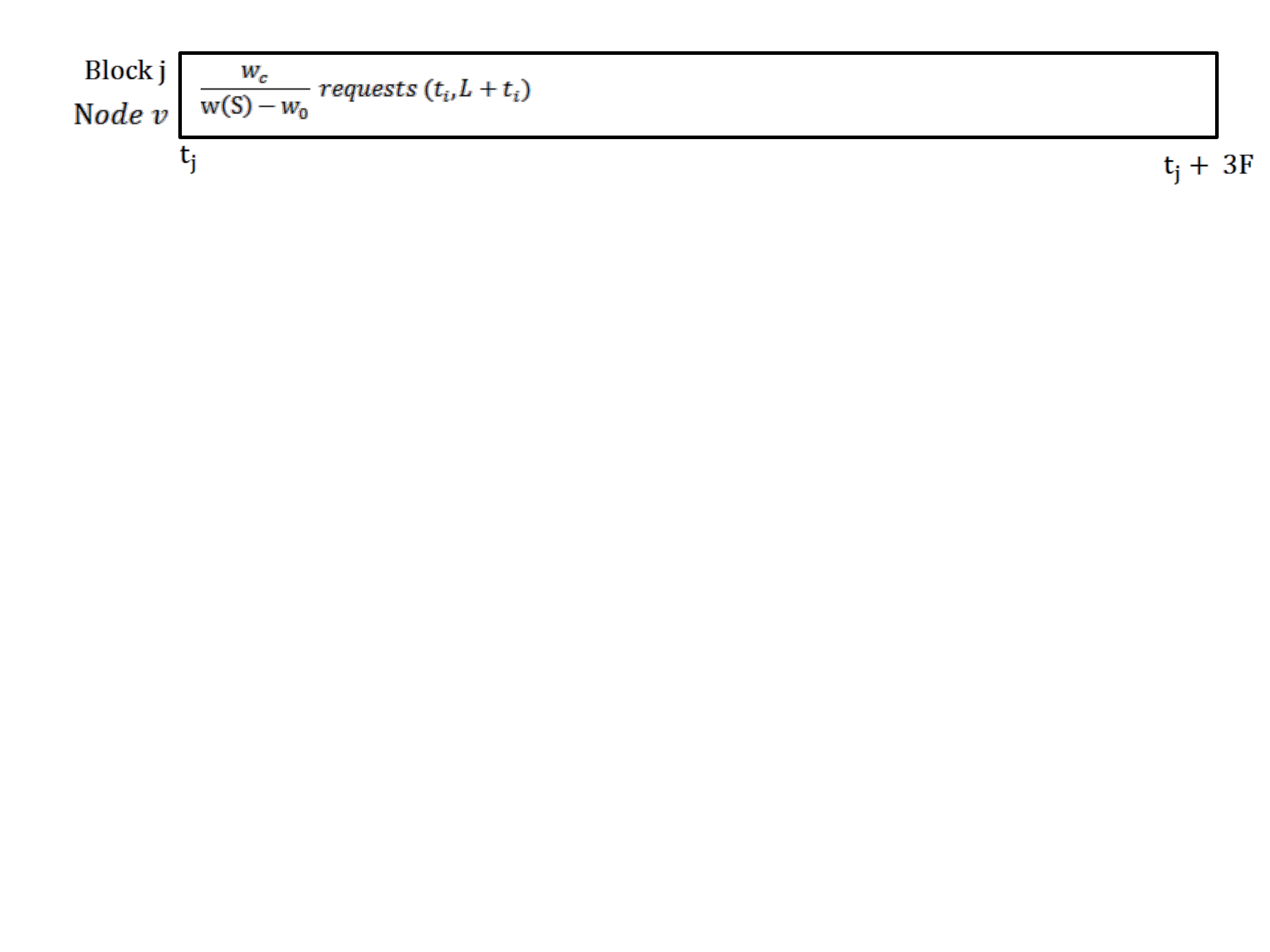}
}
\end{center}
\vspace{-190pt} \caption{Block's structure. The pair $(r,d)$
represent release time $r$ and deadline $d$. Note that all the
requests arrive at once in the beginning of the block.}
\label{figure:Block}
\end{figure}

\begin{figure}[htbp]
\begin{center}
\scalebox{0.40}
{
\includegraphics{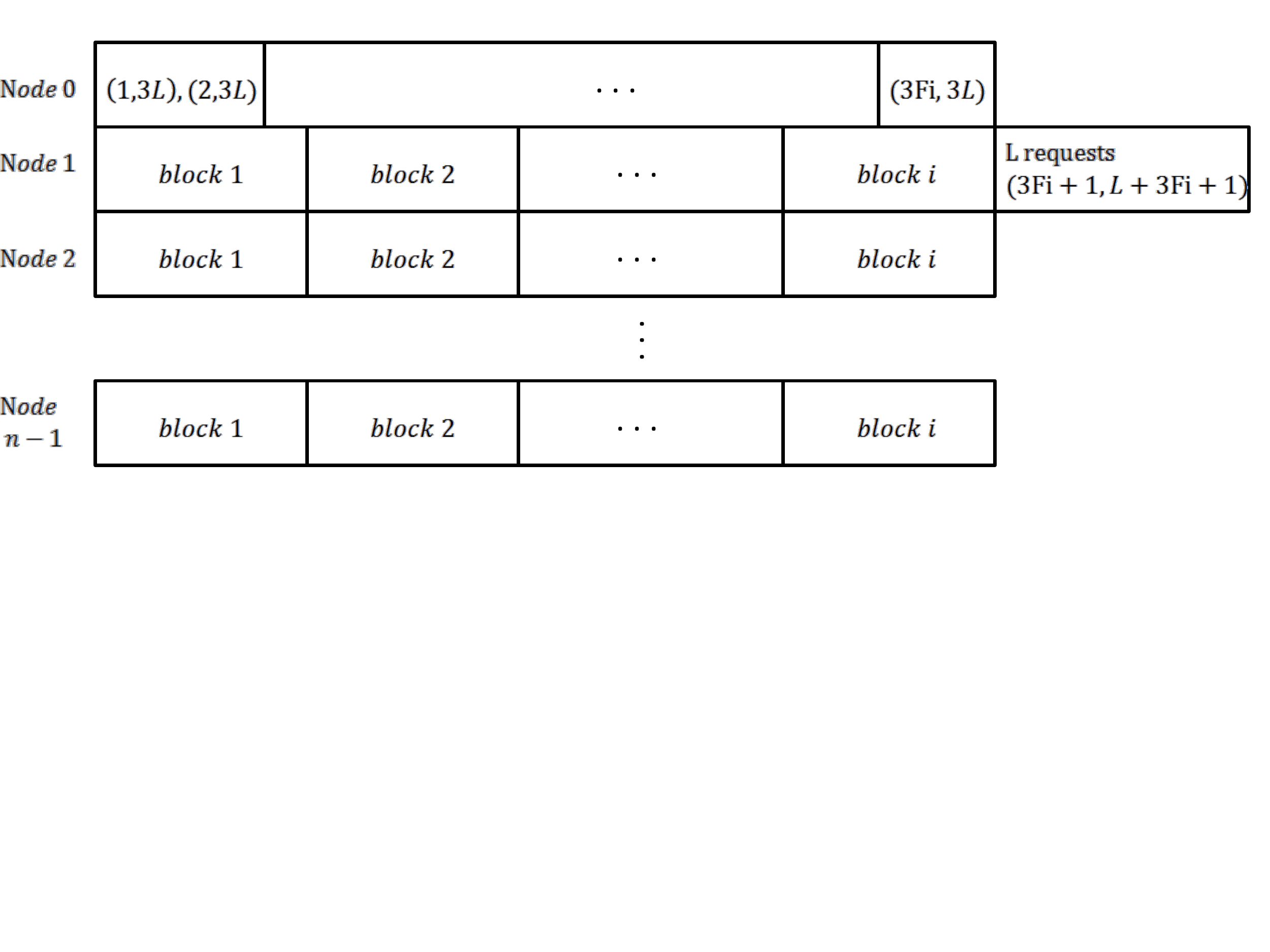}
}
\end{center}
\vspace{-110pt}
\caption{Sequence structure for Termination Case 1. See Figure \ref{figure:Block}
for blocks structure.}
\label{figure:Case1}
\end{figure}

\begin{figure}[htbp]
\begin{center}
\scalebox{0.40}
{
\includegraphics{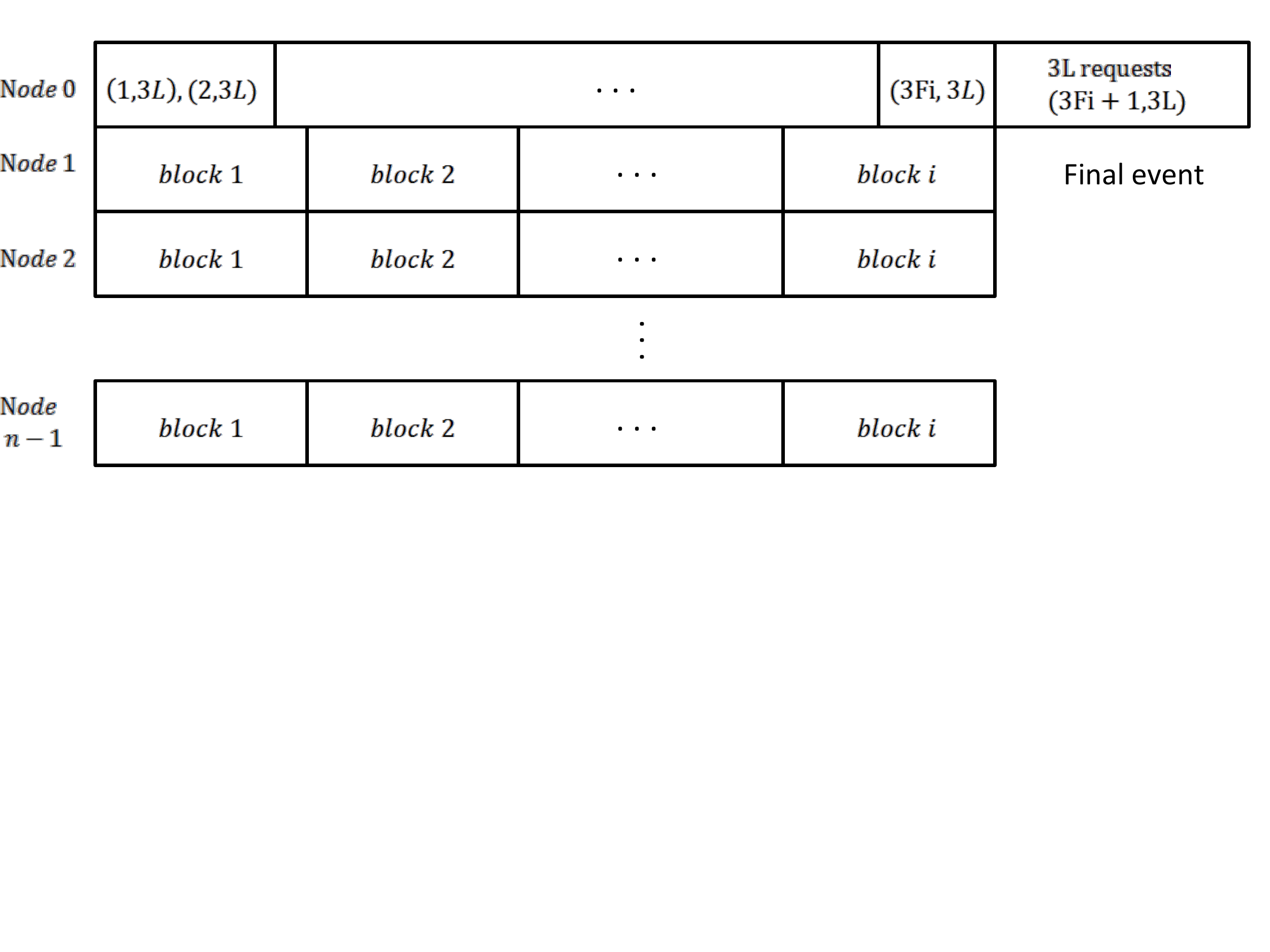}
}
\end{center}
 \vspace{-120pt}
\caption{Sequence structure for Termination Case 2. See Figure \ref{figure:Block}
for blocks structure.}
\label{figure:Case2}
\end{figure}

\begin{figure}[htbp]
\begin{center}
\scalebox{0.40}
{
\includegraphics{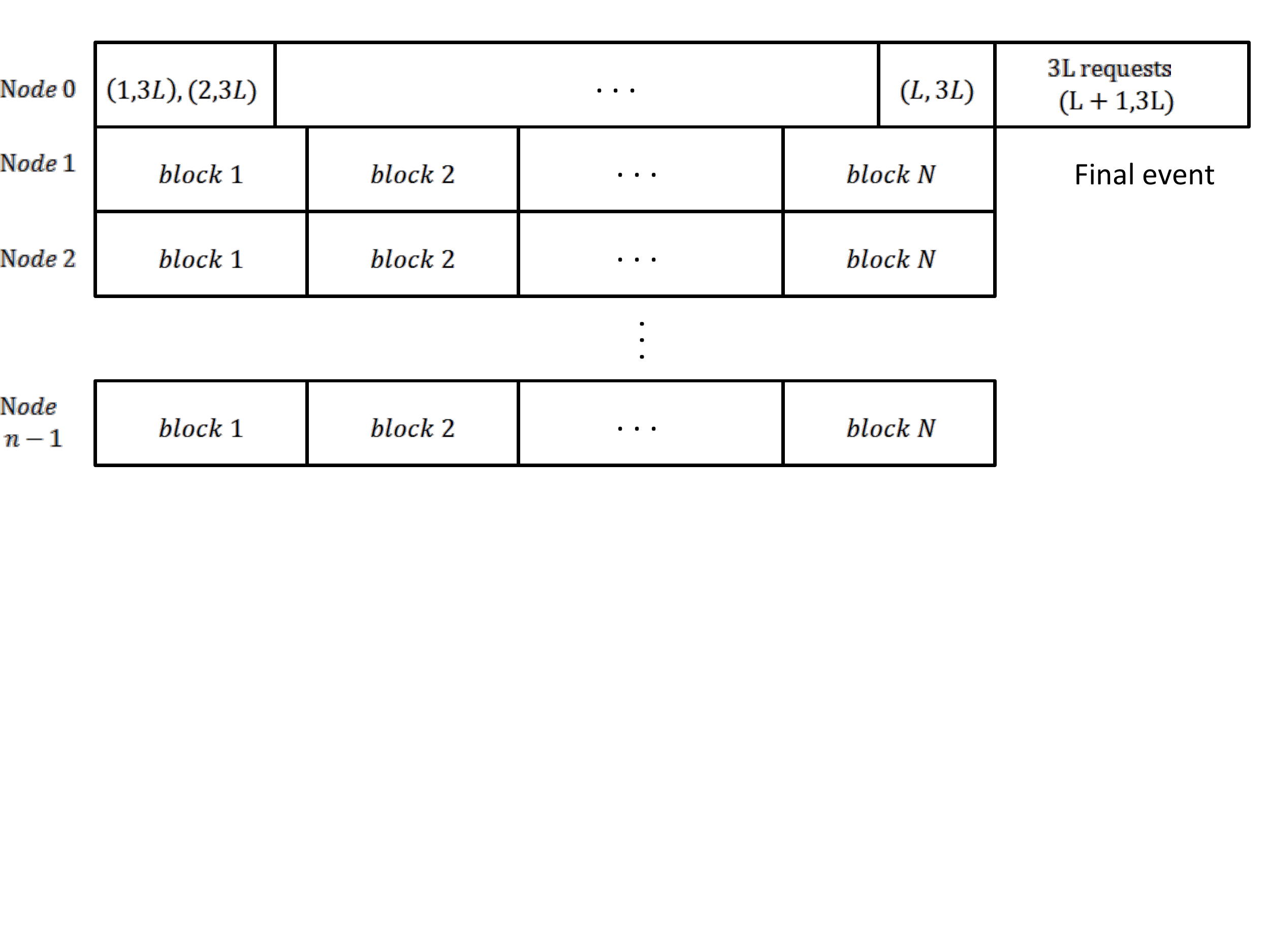}
}
\end{center}
\vspace{-120pt}
\caption{Sequence structure for Termination Case 3. Recall that $N = \frac{1}{3}\sqrt{\frac{L}{w({\rm S})}}$.
See Figure \ref{figure:Block} for blocks structure.}
\label{figure:Case3}
\end{figure}

\end{document}